\documentclass[twoside]{llncs}
\usepackage{amsmath}
  \usepackage{amssymb}
   \usepackage{graphics}
    \usepackage[dvips]{epsfig}
   \usepackage{times}
\usepackage[colorlinks=true,%
            linkcolor=blue,%
            citecolor=blue,%
            urlcolor=blue]{hyperref}           

\usepackage{fancyhdr}           
\usepackage{ifthen}

\newcommand\settitle[2][]{%
 \title{#2}
 \ifthenelse{\equal{#1}{}}%
  {\fancyhead[RO]{\nouppercase #2 \qquad \thepage}}%
  {\fancyhead[RO]{\nouppercase #1 \qquad \thepage}}%
}

\newcommand\setauthors[2]{%
 \author{#2}
  {\fancyhead[LE]{\thepage \qquad \nouppercase #1}}%
}

\fancypagestyle{plain}{%
\fancyhf{}

}

\def\keywordsname{Keywords.}
\newenvironment{keywords}{%
      \list{}{\advance\topsep by-0.50cm\relax\small
     \leftmargin=1cm
      \labelwidth=1cm
     \listparindent=1cm
     \itemindent\listparindent
      \rightmargin\leftmargin}\item[\hskip\labelsep
                                    \bfseries\keywordsname]}

   {\endlist}

   \def\R{{\mathbb R}}
   \def\N{{\mathbb N}}

   \def\Z{{\mathbb Z}}

   \def\O{{\Omega}}

   \def\qed{\hfill$\diamond$}

   \def\({\langle}
   \def\){\rangle}
   \def\mb{\boldsymbol}

   \def\cB{{\mathcal B}}

   \def\cF{{\mathcal F}}
   \def\cG{{\mathcal G}}
   
   \def\cI{{\mathcal I}}

   \def\cP{{\mathcal P}}

   \def\1{\mb1}

   \def\v0{{\bf 0}}

\begin{document}

\settitle[The ergodic decomposition of asymptotically mean stationary random sources]
         {The ergodic decomposition of asymptotically mean stationary random sources}

\setauthors{A. Sch\"onhuth}
           {Alexander Sch\"onhuth, Member IEEE}

\institute{Pacific Institute for the Mathematical Sciences\\
           School of Computing Science\\
           Simon Fraser University\\
           8888 University Drive\\
           Burnaby, BC, V5A 1S6, Canada\\
\email{schoenhuth@cs.sfu.ca}}

\date{}
\maketitle

\thispagestyle{plain}
\begin{abstract}
It is demonstrated how to represent asymptotically mean stationary
(AMS) random sources with values in standard spaces as mixtures of
ergodic AMS sources. This an extension of the well known decomposition
of stationary sources which has facilitated the generalization of
prominent source coding theorems to arbitrary, not necessarily
ergodic, stationary sources. Asymptotic mean stationarity generalizes
the definition of stationarity and covers a much larger variety of
real-world examples of random sources of practical interest.  It is
sketched how to obtain source coding and related theorems for
arbitrary, not necessarily ergodic, AMS sources, based on the
presented ergodic decomposition.
\end{abstract}


\section{Introduction}
\label{sec:introduction}

The main purpose of this paper is to demonstrate how to decompose
asymptotically mean stationary (AMS) random sources into ergodic AMS
sources.  The issue was brought up in \cite{Gray84}, as it is involved
in a variety of aspects of substantial interest to information theory.
To the best of our knowledge, it had remained unsolved since then.\par

The ergodic decomposition of AMS sources can be viewed as an extension
of the ergodic decomposition of stationary sources which states that a
stationary source can be decomposed into ergodic components or, in
other words, that it is a mixture of stationary and ergodic sources.
This was originally discussed in more abstract measure theoretic
settings (see the subsequent remark~\ref{rem.statdecomp}).\par The
first result in information theory that builds on the idea of
decomposing a source into ergodic components was obtained by Jacobs in
1963.  He proved that the entropy rate of a stationary source is the
average of the rates of its ergodic components \cite{Jacobs63}.  In
1974, the ergodic decomposition of stationary sources was rigorously
introduced to the community by Gray and Davisson~\cite{Gray74} who
also provided an intuitive proof for sources with values in a discrete
alphabet.  This turned out to be a striking success as prominent
theorems from source coding theory and related fields could be
extended to arbitrary, not necessarily ergodic, stationary sources
\cite{Gray74a,Kieffer75,Neuhoff75,Pursley76,Leon79,Effros94} (see the
references therein as well as \cite{Gray90} for a complete list).\par
In general, these results underscore that ergodic and information
theory have traditionally been sources of mutual inspiration.\\ 

\begin{remark}
\label{rem.statdecomp}
The first variant of an ergodic decomposition of stationary sources
(with values in certain topological spaces) was elaborated in a
seminal paper by von Neumann~\cite{vonNeumann}.  Subsequently, Kryloff
and Bogoliouboff \cite{Bogoliouboff} obtained the result for compact
metric spaces.  and it was further extended by Halmos
\cite{Halmos41,Halmos49} to normal spaces.  In parallel,
Rokhlin~\cite{Rokhlin52} proved the decomposition theorem for Lebesgue
spaces, which still can be considered as one of the most general
results. Oxtoby \cite{Oxtoby52} further clarified the situation by
demonstrating that Kryloff's and Bogoliouboff's results can be
obtained as corollaries of Riesz' representation theorem. In ergodic
theory, the corresponding idea is now standard
\cite{Pollicott,Walters}.
\end{remark}

Asymptotic mean stationarity was first introduced in 1952 by Dowker
\cite{Dowker51} and further studied by Rechard \cite{Rechard56}, but
became an area of active research only in the early 1980s, thanks to a
fundamental paper of Gray and Kieffer~\cite{Gray80}. Asymptotic mean
stationarity is a property that applies for a large variety of natural
examples of sources of practical interest \cite{Gray80}. Reasons are:
\begin{enumerate}
\item Asymptotic mean stationarity is stable under conditioning (see
  \cite{Krengel}, p.~33) whereas stationarity is not.
\item To possess ergodic properties w.r.t. bounded
  measurements is equivalent to asymptotic mean stationarity
  \cite{Dowker51,Gray80}.  Note that Birkhoff's theorem
  (e.g.~\cite{Krengel}) states that stationarity is sufficient to
  possess ergodic properties.
\item The Shannon-McMillan-Breiman (SMB) theorem was iteratively
  extended to finally hold for AMS discrete random sources in 1980
  \cite{Gray80}.
\end{enumerate}
Note that an alternative, elegant proof of the SMB theorem can be
achieved by employing the ergodic decomposition of stationary
sources~\cite{Algoet88}.  The second point gives evidence of the
practical relevance of AMS sources, as to possess ergodic properties
is a necessity in a wide range of real-world applications of
stochastic processes. For example, asymptotic mean stationarity is
implicitly assumed when relative frequencies along sequences emitted
by a real-world process are to converge. See also
\cite{Kieffer81,Faigle07} for expositions of large classes of AMS
processes of practical interest.  The validity of the SMB theorem is
a further theoretical clue to the relevance of AMS sources in
information theory.\\ 

The benefits of an ergodic decomposition of AMS sources are, on one
hand, to arrange the theory of AMS sources and, on the other hand, to
facilitate follow-up results in source coding theory and related
fields (see the discussion section~\ref{sec.discussion} for some
immediate consequences). In \cite{Gray84}, one can find a concise
proof of the ergodic decomposition of stationary sources as well as
the ergodic decomposition of two-sided AMS sources, both with values
in standard spaces.  The case of two-sided AMS sources, however, is a
straightforward reduction to the stationary case which does not apply
for arbitrary AMS sources.  As the result for arbitrary AMS sources
would have been highly desirable, it was listed as an open question in
the discussion section of \cite{Gray84}.\par The main purpose of this
paper is to provide a proof of the ergodic decomposition of arbitrary
(two-sided and one-sided) AMS sources with values in standard spaces
which cover discrete-valued and all natural examples of topological
spaces.\\

The paper is organized as follows. In section~\ref{sec.results} we
collect basic notations and state the two main results. The first one
is the ergodic decomposition itself and the second one is an essential
lemma that may be interesting in its own right. In
section~\ref{sec.conv}, we present basic definitions of probability
and measure theory as well as a classical ergodic theorem
(\emph{Krengel's stochastic ergodic theorem}) required for our
purposes. The statement of Krengel's theorem is intuitively easy to
grasp and can be understood by means of basic definitions from
probability theory only.  In section~\ref{sec.amsconv} we give a proof
of lemma~\ref{l.amsconv}. Both the statement and the proof of
lemma~\ref{l.amsconv} are crucial for the proof of the
decomposition. In section~\ref{sec.standardcond}, we list relevant
basic properties of standard spaces (subsection~\ref{ssec.standard})
and regular conditional probabilities and conditional expectations
(subsection~\ref{ssec.regular}). Finally, in section~\ref{sec.decomp},
we present the proof of the ergodic decomposition. For organizational
convenience, we have subdivided it into three steps and collected the
merely technical passages into lemmata which have been deferred to the
appendices \ref{app.lemma} and \ref{app.proof2}. We conclude by
outlining immediate consequences of our result and pointing out
potential applications in source coding theory, in the discussion
section~\ref{sec.discussion}.

\section{Basic Notations and Statement of Results}
\label{sec.results}

Let $(\O,\cB)$ be a measurable space and $T:\O\to\O$ a measurable function. In
this setting (see \cite{Pollicott,Gray01}), a probability measure $P$ is
called {\em stationary} (relative to $T$), if
\begin{equation*}
P(B) = P(T^{-1}B)
\end{equation*}
for all $B\in\cB$. It is called {\em asymptotically mean stationary (AMS)}
(relative to $T$), if there is a measure $\bar{P}$ on $(\Omega,\mathcal{B})$
such that
\begin{equation}
\forall B\in\mathcal{B}:\quad
\lim_{n\to\infty}\frac{1}{n}\sum_{i=0}^{n-1}P(T^{-i}B) 
= \bar{P}(B).\label{eq.ams}
\end{equation}
Clearly, the measure $\bar{P}$ is stationary and it is therefore called the
{\em stationary mean} of $P$. An event $I\in\cB$ is called
\emph{invariant} (relative to $T$), if $T^{-1}I = I$.  The set of invariant
events is a sub-$\sigma$-algebra of $\cB$ which we will denote by $\cI$.  A
probability measure $P$ on $(\O,\cB)$ is said to be \emph{ergodic} (relative
to $T$), if $P(I)\in\{0,1\}$ for any such invariant $I\in\cI$.  Note that an
AMS system is ergodic if and only if its stationary mean is.\par In order to
apply this theory to {\it ($A$-valued) random sources}, that is, discrete-time
stochastic processes with values in a standard space $A$ (for a definition of
standard space see subsection~\ref{ssec.standard}), one sets
\begin{equation*}
\O = A^I = \bigotimes_{i\in I}A
\end{equation*}
where $I \in\{\N,\Z\}$. That is, $\O$ is the space of one-sided ($I=\N$) or two-sided
($I=\Z$) $A$-valued sequences. $\cB$ then is
set to be the $\sigma$-algebra generated by the cylinder sets of sequences. A
random source is given by a probability measure $P$ on
$(\O,\cB)$. Further, $T:\O\to\O$ is defined to be the {\it left shift
operator}, i.e.
\begin{equation*}
(Tx)_n = x_{n+1}
\end{equation*}
for $x=(x_0,x_1,...,x_n,...)\in\O$ (one-sided case) or $x =
(...,x_{-1},x_0,x_1,...)\in\O$ (two-sided case).\\

The main contribution of this paper is to give a proof of the
following theorem.

\begin{theorem}
\label{t.amsdecomp}
Let $P$ be a probability measure on a standard space $(\Omega,\cB)$ 
which is AMS relative to the measurable $T:\O\to\O$. 
Then there is a $T$-invariant set $E\in\cI$ with $P(E)=1$ such
that for each $\omega\in E$ there is an {\bf ergodic AMS}
probability measure $P_{\omega}$ and the following properties
apply:
\begin{enumerate}
\item[(a)] \begin{equation*}
           \forall B\in\cB : \quad P_{\omega}(B) = P_{T\omega}(B).
	   \end{equation*}
\item[(b)] 
      \begin{equation*}
      \forall B\in\cB: \quad P(B) = \int P_{\omega}(B)\,dP(\omega).
      \end{equation*} 
      
\item[(c)] If $f\in L_1(P)$, then also $\omega\mapsto\int f\,dP_{\omega}\in L_1(P)$ and 
      \begin{equation*}
      \int f\,dP(\omega) = \int (\int f\,dP_{\omega})\,dP(\omega).
      \end{equation*}
\end{enumerate}
\end{theorem}
\medbreak

Replacing AMS by stationary yields the aforementioned and well-known theorem
of the ergodic decomposition of stationary random sources
(e.~g.~\cite{Gray84}, th.~2.5).\\ 

The following lemma is a key observation for the proof of
theorem~\ref{t.amsdecomp} and may be interesting in its own right. It
states that the convergence involved in the definition of AMS measures
is uniform over the elements of $\cB$.  This may seem intuitively
surprising, as the underlying measurable space does not even have to
be standard.

\begin{lemma}
\label{l.amsconv}
Let $P$ be an AMS measure on $(\O,\cB)$ relative to $T$. 
Then 
\begin{equation*}
\sup_{B\in\cB}\;|\frac1n\sum_{i=0}^{n-1}P(T^{-i}B) - \bar{P}(B)| \underset{n\to\infty}{\longrightarrow} 0.
\end{equation*}
In other words, the convergence of (\ref{eq.ams}) is uniform
over the events $B\in\cB$.
\end{lemma}


\section{Preliminaries}
\label{sec.conv}

\subsection{Convergence of Measures}
\label{ss.conv}

\begin{definition}\label{d.measureconv}
Let $(P_n)_{n\in\mathbb{N}}$ be a sequence of probability measures
on a measurable space $(\Omega,\mathcal{B})$. 
\begin{itemize}
\item We say that the $P_n$ converge {\bf strongly} to a probability measure
  $\bar{P}$ if the sequences $(P_n(B))_{n\in\mathbb{N}}$ converge to
  $\bar{P}(B)$ for all $B\in\mathcal{B}$.
\item If this convergence happens to be uniform in $B\in\mathcal{B}$ we 
say that the $P_n$ converge {\bf Skorokhod weakly} to $\bar{P}$.
\end{itemize}
\end{definition}

See \cite{Jacka97} for history and detailed characterisations of these
definitions.  Obviously Skorokhod weak convergence implies strong
convergence. Seen from this perspective, lemma~\ref{l.amsconv} states that the
measures $P_n=1/n\sum_{t=0}^{n-1}P\circ T^{-t}$, where $P$ is an AMS measure
and $P\circ T^{-t}(B):=P(T^{-t}B)$ , do not only converge strongly (which they
do by definition), but also Skorokhod weakly to the stationary mean $\bar{P}$.\\

A helpful characterization of Skorokhod weak convergence is the
following theorem. Therefore we recall that a probability measure $Q$
is said to {\it dominate} another probability measure $P$ (written
$Q>>P$) if $Q(B)=0$ implies $P(B)=0$ for all $B\in\cB$. The theorem of
Radon-Nikodym (e.g.~\cite{Halmos}) states that in case of $Q>>P$ there
is a measurable function $f:\O\to\R$, called {\it Radon-Nikodym
  derivative} or simply {\it density}, written $f=\frac{dP}{dQ}$, such
that
\begin{equation*}
P(B) = \int_Bf\,dQ
\end{equation*}
for all $B\in\cB$. It holds that $P(f=g) = 1$ (hence $Q(f=g)=1$) for
two densities $f,g = \frac{dP}{dQ}$.\par 
As usual, 
\begin{equation*}
L_1(Q):=L_1(\O,\cB,Q)
\end{equation*}
denotes the (linear) space of $Q$-integrable functions on $(\O,\cB)$
modulo the subspace of functions that are null almost everywhere. For
technical convenience, we will sometimes identify elements of $f\in
L_1(Q)$ with their representatives $f:\Omega\to\R$. As a consequence
we have that $f=g$ in $L_1(Q)$ if and only if $Q(f=g) = 1$ for their
representatives. That is, equality is in an almost-everywhere sense
for the representatives. Therefore, in $L_1(Q)$, a density is
unique. Furthermore, $L_1(Q)$ can be equipped with the norm
\begin{equation*}
||f||_1 := \int_{\O} |f|\,dQ.
\end{equation*}
See standard textbooks (e.g.~\cite{Halmos}) for details.\\ 

In this language, Skorokhod weak convergence has a useful
characterisation.

\begin{theorem}[\cite{Jacka97}]
\label{t.jacka}
Let $(P_n)_{n\in\N},\bar{P}$ be probability measures.  Then the
following statements are equivalent:
\begin{enumerate}
\item[(i)] The $P_n$ converge Skorokhod weakly to $\bar{P}$.
\item[(ii)] There is a probability measure $Q$, which dominates $\bar{P}$ and
all of the $P_n$ such that the densities $f_n:=\frac{dP_n}{dQ}$ converge {\bf
stochastically} to the density $\bar{f}:=\frac{d\bar{P}}{dQ}$, that is
\begin{equation*}
\forall\epsilon\in\R^+:\quad Q(\{\omega : |f_n(\omega) - \bar{f}(\omega)| > \epsilon\})
\quad\underset{n\to\infty}{\longrightarrow}\quad 0.
\end{equation*}
\item[(iii)] There is a probability measure $Q$, which dominates $\bar{P}$ and
all of the $P_n$ such that the densities $f_n:=\frac{dP_n}{dQ}$ converge {\bf
in mean} (in $\mathbf{L_1(Q)}$) to the density $\bar{f}:=\frac{d\bar{P}}{dQ}$,
that is
\begin{equation*}
\int |f_n - \bar{f}|\,dQ \quad\underset{n\to\infty}{\longrightarrow}\quad 0.
\end{equation*}
\end{enumerate}
\end{theorem}

\begin{proof} See \cite{Jacka97}, pp.~6--7.\qed\end{proof}
\medbreak

\subsection{Krengel's theorem}
\label{ss.krengel}

In few words, the {\em stochastic ergodic theorem} of Krengel states that the
averages of densities which are obtained by iterative applications of a {\em
positive contraction} in $L_1(Q)$ converge stochastically to a density
that is invariant with respect to the positive contraction.

To be more precise, let $(\Omega,\mathcal{B},P)$ be a measure space and $U$ a
positive contraction on $L_1(\Omega,\mathcal{B},P)$, that is, $Uf\ge 0$ for
$f\ge 0$ (positivity) and $||Uf||_1 \le ||f||_1$ (contraction).  Then $\Omega$
can be decomposed into two disjoint subsets (uniquely determined up to
$P$-nullsets)
\begin{equation*}
\Omega = \tilde{C}\;\dot{\cup}\;\tilde{D},
\end{equation*}
where $\tilde{C}$ is the maximal support of a $f_0\in
L_1(\Omega,\mathcal{B},P)$ with $Uf_0 = f_0$. In other words, for all $f\in
L_1$ with $Uf = f$, we have $f=0$ on $\tilde{D}$ and there is a $f_0\in L_1$
such that both $Uf_0=f_0$ and $f_0>0$ on $\tilde{C}$ (see \cite{Krengel},
p. 141 ff. for details). Krengel's theorem then reads as follows.

\begin{theorem}[Stochastic ergodic theorem; Krengel]
\label{t.krengel}
If $U$ is a positive contraction on $L_1$ of a $\sigma$-finite measure space
$(\Omega,\mathcal{B},Q)$ (e.g. a probability space, the definition of a
$\sigma$-finite measure space \cite{Halmos} is not further needed here) then,
for any $f\in L_1$, the averages
\begin{equation*}
A_nf := \frac{1}{n}\sum_{t=0}^{n-1}U^tf
\end{equation*}
converge stochastically to a $U$-invariant $\bar{f}$. Moreover, on
$\tilde{C}$ we have $L_1$-convergence, whereas on $\tilde{D}$ the
$A_nf$ converge stochastically to $0$. If $f\ge 0$ then 
\begin{equation}\label{eq.krengelliminf}
f=\liminf_{n\to\infty}A_nf\quad\text{ in }L_1(Q).
\end{equation}
\end{theorem}

\begin{proof} 
\cite{Krengel}, p.143.\qed\end{proof}
\medbreak

\subsection{Finite Signed Measures}
\label{ss.sm}

Let $(\Omega,\mathcal{B})$ be a measurable space. A finite signed measure is a
$\sigma$-additive, but not necessarily positive, finite set function on
$\cB$. The theorem of the Jordan decomposition (\cite{Halmos}, p.~120 ff.)
states that $P = P_+ - P_-$ for measures $P_+,P_-$. These measures are
uniquely determined insofar as if $P = P_1 - P_2$ for measures $P_1,P_2$ then
there is a measure $\delta$ such that 
\begin{equation}\label{eq.jordan}
P_1 = P_+ + \delta\quad\text{ and }\quad P_2 =
P_-+\delta.
\end{equation}
$P_+,P_-$ and $|P|:=P_++P_-$ are called {\em positive, negative}
and {\em total variation} of $P$. We further define
\begin{equation*}
||P||_{TV} := |P|(\O).
\end{equation*}
By ``eventwise'' addition and scalar multiplication the set of finite
signed measures can be made a normed vector space equipped with the
{\em norm of total variation} $||.||_{TV}$, written $(\cP,||.||_{TV})$
or simply $\cP$. The following observation about signed measures and
measurable functions is crucial for this work.

\begin{lemma}\label{l.ttv}
Let $P$ be a finite signed measure on $(\O,\cB)$ and $T:\O\to\O$
a measurable function. Then $P\circ T^{-1}$ is a finite signed measure
for which
\begin{equation*}
|P\circ T^{-1}|(B) \le |P|(T^{-1}B)
\end{equation*}
for all $B\in\cB$. In particular, $||P\circ T^{-1}||_{TV} \le
||P||_{TV}$.
\end{lemma}

\begin{proof}
Note that $P\circ T^{-1} = P_+\circ T^{-1} - P_-\circ
T^{-1}$ is a decomposition into a difference of measures. Because of the
uniqueness property of the Jordan decomposition (\ref{eq.jordan}), there is a
measure $\delta$ such that $P_+\circ T^{-1} = (P\circ T^{-1})_+ +\delta$ and
$P_-\circ T^{-1} = (P\circ T^{-1})_- + \delta$. Therefore $|P\circ T^{-1}|(B) =
(P\circ T^{-1})_+(B) + (P\circ T^{-1})_-(B)\le P_+(T^{-1}B) + P_-(T^{-1}B) =
|P|(T^{-1}B)$. $B=\O$ yields the last assertion, as $T^{-1}\O = \O$.\qed
\end{proof}
\medbreak

We finally observe the following well known relationship between signed
measures dominated by a measure $Q$ and $L_1(Q)$. Therefore, as usual
(e.g.~\cite{Halmos}), we say that a finite, signed measure $P$ is dominated by
$Q$ if its total variation is, that is, $|P|<<Q$. Note that the set $\cP_Q$ of
finite, signed measures that are dominated by $Q$ is a linear subspace of
$\cP$.

\begin{lemma}\label{l.sml1}
Let $Q$ be a measure on the measurable space $(\O,\cB)$ and $\cP_Q$ be the
linear space of the finite signed measures that are dominated by $Q$. If
$P_f(B) := \int_Bf\,dQ$ for $f\in L_1(Q)$, then
\begin{equation*}
\begin{array}{rccc}
\Phi:& (L_1(Q),||.||_1) & \longrightarrow& (\cP_Q,||.||_{TV})\\
     & f     & \mapsto        & P_f
\end{array}
\end{equation*}
establishes an isometry of normed vector spaces. 
\end{lemma} 

\begin{proof}
This is a consequence of the theorem of Radon-Nikodym, see
\cite{Halmos}, p.~128 ff.  If $P$ is a finite signed measure with 
$|P|<<Q$ then also $P_+,P_-<<Q$. Define $\Psi(P):=\frac{dP_+}{dQ} -
\frac{dP_-}{dQ}\in L_1(Q)$ as the difference of the densities of
$P_+,P_-$ relative to $Q$. Then $\Psi$ is just the inverse of $\Phi$.
It is straightforward to check that $||f||_1 = ||\Phi(f)||_{TV}$.\qed\end{proof}
\medbreak

\section{Proof of Lemma \ref{l.amsconv}}
\label{sec.amsconv}

We start by illustrating one of the core techniques of this work.  Let
$(\Omega,\cB)$ be a measurable space and $(Q_n)_{n\in\N}$ be a
countable collection of probability measures on it. Then the set
function defined by
\begin{equation}
\label{eq.Qa}
Q(B) := \sum_{n\ge 0}2^{-n-1}Q_n(B)\quad\forall B\in\cB
\end{equation}
is a probability measure which dominates all of the $Q_n$ \cite{Jacka97}.\par
Let now $(\O,\cB,P,T)$ be such that $P$ is an AMS measure
relative to the measurable $T:\O\to\O$. Define further $P_n$ to be the
measures given by
\begin{equation}\label{eq.pn}
P_n(B)=\frac1n\sum_{t=0}^{n-1}P(T^{-t}B)
\end{equation}
for $B\in\cB$. As a consequence of (\ref{eq.Qa}), the set function $Q$
defined by
\begin{equation}\label{eq.Q}
Q(B) := \frac12(\bar{P}(B) + \sum_{n\ge 0}2^{-n-1}P(T^{-n}B))
\end{equation}
for $B\in\cB$ is a probability measure which dominates all of the
$P\circ T^{-n}$ as well as $\bar{P}$. Hence it also dominates
all of the $P_n$. Accordingly, we write
\begin{equation}\label{eq.fn}
f_n :=\frac{dP_n}{dQ} \quad\text{ and }\quad \bar{f}:=\frac{\bar{P}}{dQ}
\end{equation}
for the respective densities. Lemma~\ref{l.amsconv} can be obtained as a
corollary of the following result.

\begin{lemma}\label{l.amsconv2}
Let $P$ be an AMS probability measure on $(\Omega,\cB)$ relative to
$T$ with stationary mean $\bar{P}$. Let $P_n$, $Q$, $f_n$ and
$\bar{f}$ as defined by equations (\ref{eq.pn}),(\ref{eq.Q}) and
(\ref{eq.fn}). Then the $f_n$ converge stochastically to the density
$\bar{f}:=\frac{d\bar{P}}{dQ}$.  Moreover,
\begin{equation}\label{eq.liminffn}
\bar{f} = \liminf_{n\to\infty}f_n\quad \text{$Q$-a.e.}
\end{equation}
\end{lemma}

\medskip\begin{proof}
Let $f_1=\frac{dP}{dQ}$. The road map of the proof is to
construct a positive contraction $U$ on $L_1(Q)$ such that
\begin{equation*}
f_n = \frac1n\sum_{t=0}U^tf_1 =: A_nf_1.
\end{equation*}
As a consequence of Krengel's theorem we will obtain that the
$f_n$ converge stochastically to a $U$-invariant limit
$f^*$. In a final step we will show that indeed $f^* = \bar{f}$
in $L_1(Q)$ (i.e. $Q$-a.e.), which completes the proof.
\medbreak

Our endomorphism $U$ on $L_1(Q)$ is induced by the measurable function
$T$. Let $f\in L_1(Q)$. We first recall that, by lemma \ref{l.sml1}, the
set function $\Phi(f)$ given by
\begin{equation*}
\Phi(f)(B) := \int_Bf\,dQ
\end{equation*}
for $B\in\cB$ and $f\in L_1(Q)$ is a finite, signed measure 
on $(\Omega,\mathcal{B})$ whose total variation $|\Phi(f)|$ is 
dominated by $Q$.
\medbreak

We would like to define
\begin{equation*}
Uf:=\Phi^{-1}(\Phi(f)\circ T^{-1}),
\end{equation*}
which would be obviously linear. However, $\Phi^{-1}$ is only defined
on $\cP_Q$, that is, for finite signed measures that are dominated by
$Q$. Therefore, we have to show that $\Phi(f)\circ T^{-1}\in\cP_Q$
which translates to demonstrating that $|\Phi(f)\circ
T^{-1}|<<Q$. This does not hold in general (see
\cite{Krengel}). However, in the special case of the dominating $Q$
chosen here, it can be proven.
\medbreak

To see this let $B$ such that $|\Phi(f)\circ T^{-1}|(B) > 0$ and we have to
show that $Q(B)>0$.  Because of lemma \ref{l.ttv}  
\begin{equation*}
|\Phi(f)|(T^{-1}B) \ge |\Phi(f)\circ T^{-1}|(B) > 0.
\end{equation*}
As $|\Phi(f)| << Q$, we obtain $Q(T^{-1}B) > 0$. By definition of $Q$ we thus
either find an $N_0\in\mathbb{N}$ such that $0 < P(T^{-N_0}(T^{-1}B)) =
P(T^{-N_0-1}B)$ or we have that $0 < \bar{P}(T^{-1}B) = \bar{P}(B)$ because of
the stationarity of $\bar{P}$. Both cases imply $Q(B) >
0$ which we had to show.  \medbreak

If $f\ge 0$ then $\Phi(f)$ is a measure. Hence also $\Phi(f)\circ T^{-1}$ is
a measure which in turn implies $Uf=\frac{d(\Phi(f)\circ T^{-1})}{dQ·}\ge 0$.  
Hence $U$ is positive.
It is also a contraction with respect to the $L_1$-norm $||.||_1$, as,
because of the lemmata \ref{l.ttv} and \ref{l.sml1},
\begin{equation*}
||Uf||_1 = ||\Phi(f)\circ T^{-1}||_{TV} \le ||\Phi(f)||_{TV} = ||f||_1.
\end{equation*}
\medbreak

For $f_1 = \frac{dP}{dQ}$ being the density of $P$ relative to $Q$ we
obtain
\begin{equation*}
U^nf_1 = \frac{d(P\circ T^{-n})}{dQ}
\end{equation*}
Hence the $f_n:=A_nf_1 = 1/n\sum_{t=0}^{n-1}U^tf_1$ are the densities
of the $P_n=\frac1n\sum_{t=0}^{n-1}P\circ T^{-t}$ relative to $Q$.  An
application of Krengel's theorem \ref{t.krengel} then shows that the
$A_nf_1$ converge stochastically to a $U$-invariant limit $f^*\in
L_1(Q)$. Note that a positive $U$-invariant $f$ just corresponds to a
stationary measure.  \medbreak

It remains to show that $\bar{f} = f^*$ in $L_1(Q)$ or, equivalently,
$\bar{f} = f^*$ $Q$-a.e. for their representatives (see the
discussions in subsection~\ref{ss.conv}). Let $\tilde{D}$, as
described in subsection~\ref{ss.krengel}, be the complement of the
maximal support of a $U$-invariant $g\in L_1(Q)$. We recall that
stationary measures are identified with positive, $U$-invariant
elements of $L_1(Q)$. Therefore, 
$\bar{f}=\frac{d\bar{P}}{dQ}$ is $U$-invariant which yields
\begin{equation*}
Q(\{\bar{f} > 0\} \;\cap \;\tilde{D}) = 0
\end{equation*}
which implies $\bar{f} = 0$ $Q$-a.e. on $\tilde{D}$. Due to Krengel's
theorem, it holds that also $f^* = 0$ $Q$-a.e. on $\tilde{D}$, and we
obtain that
\begin{equation*}
\bar{f} = 0 = f^*\quad Q-\text{a.e. on }\tilde{D}.
\end{equation*}
In order to conclude that 
\begin{equation*}
\bar{f} = f^*\quad Q-\text{a.e. on }\tilde{C}
\end{equation*}
it remains to show that $\int_Bf^*\,dQ=\int_B\bar{f}\,dQ$ for events
$B\subset\tilde{C}=\Omega\setminus\tilde{D}$ as two integrable
functions conincide almost everywhere if their integrals over
arbitrary events coincide (\cite{Halmos}) with which we will have
completed the proof. From Krengel's theorem we know that, on
$\tilde{C}$, we have $L_1$-convergence of the $f_n$:
\begin{equation}
\label{eq.l1conv}
\lim_{n\to\infty}\int_{\tilde{C}}|f_n - f^*|\,dQ = 0.
\end{equation}
Therefore, for $B\subset\tilde{C}$,
\begin{equation*}
\int_Bf^*\,dQ \stackrel{(\ref{eq.l1conv})}{=} \lim_{n\to\infty}\int_Bf_n\,dQ 
= \lim_{n\to\infty}P_n(B) \stackrel{(**)}{=} \bar{P}(B) = \int_B\bar{f}\,dQ,
\end{equation*}
where $(**)$ follows from the asymptotic mean stationarity of $P$.  We
thus have completed the proof of the main statement of the lemma.
\medbreak

Finally, (\ref{eq.liminffn}) is a direct consequence of (\ref{eq.krengelliminf})
in Krengel's theorem.\qed\end{proof}
\medbreak

In sum, we have shown that there is a measure $Q$ that dominates all of the
$P_n$ as well as $\bar{P}$ such that the densities of the $P_n$ converge
stochastically to the density of $\bar{P}$. According to
theorem~\ref{t.jacka},
this is equivalent to Skhorokhod weak convergence. Hence we obtain lemma
\ref{l.amsconv} as a corollary.

\section{Preliminaries II}
\label{sec.standardcond}

In this section we will first review a couple of additional definitions that
are necessary for a proof of theorem~\ref{t.amsdecomp}. In
subsection~\ref{ssec.standard} we give the definition of a standard space. The
beneficial properties of standard spaces become apparent in
subsection~\ref{ssec.regular}, where we shortly review conditional
probabilities and expectation.\par 

\subsection{Standard spaces}\label{ssec.standard}

See \cite{Partha05}, ch.~3 or \cite{Gray01} for thorough treatments of
standard spaces. In the following, a field $\cF$ is a collection of subsets of
a set $\O$ that contains $\O$ and is closed with respect to complements and
finite unions.

\begin{definition}\label{d.count}
A field $\cF$ on a set $\Omega$ is said to have 
the {\bf countable extension property} if the following two conditions
are met.
\begin{enumerate}
\item $\cF$ has a countable number of elements.
\item Every nonnegative and finitely additive set function $P$ on $\cF$
      is continuous at $\emptyset$, that is, for a sequence of elements
      $F_n\in\cF$ with $F_{n+1}\subset F_n$ such that
      $\cap_nF_n=\emptyset$ we have $\lim_{n\to\infty}P(F_n) = 0$.
\end{enumerate}
\end{definition}

\begin{definition}\label{d.standard}
A measurable space $(\Omega,\cB)$ is called a {\bf standard space}, if
the $\sigma$-algebra $\cB$ is generated by a field $\cF$ which has
the countable extension property. 
\end{definition}

\begin{remark} $\quad$
\begin{enumerate}
\item Most of the prevalent examples of
measurable spaces in practice are standard. For example, any
measurable space which is generated by a complete, separable, metric
space (i.e. a {\em Polish space}) is standard. Moreover, standard
spaces can be characterized as being isomorphic to subspaces $(B, \cB\cap
B)$ of Polish spaces $(\O,\cB)$ where $B\in\cB$ is a measurable set (see
\cite{Partha05}, ch.~3).
\item An alternative
characterisation of standard spaces is that the $\sigma$-algebra $\cB$
possesses a {\em basis}. See \cite{Katok}, app.~6, for a
discussion.
\end{enumerate}
\end{remark}
\medbreak

\subsection{Conditional Probability and Expectation}
\label{ssec.regular}

See \cite{Partha05}, ch.~6 or \cite{Gray01} for a discussion of conditional
probability and expectation.

\begin{definition}\label{d.regprob}
Let $P$ be a probability measure on a measurable space $(\Omega,\cB)$
and let $\cG\subset\cB$ be a sub-$\sigma$-algebra of $\cB$. A function
\begin{equation*}
\delta(.,.):\cB\times\Omega\to\R, 
\end{equation*}
is called a (version of the) {\bf conditional probability} of $P$
given $\cG$, if
\begin{description}
\item[(CP1)] $\delta(B,.)$ is $\cG$-measurable for all $B\in\cB$ and
\item[(CP2)] \begin{equation*}
	    P(B\cap G) = \int_G \delta(B,\omega)\,dP(\omega)
            \end{equation*}
            for all $G\in\cG,B\in\cB$.
\end{description}
$\delta(.,.)$ is called a (version of the) {\bf regular conditional
probability} of $P$ given $\cG$, if, in addition to (CP1) and (CP2),
\begin{description}
\item[(RCP)] $\delta(.,\omega)$ is a probability measure on $\cB$ for all $\omega\in\Omega$.
\end{description}
\end{definition}

We collect a couple of basic results about conditional probabilities. See
\cite{Partha05} or \cite{Gray01} for details.
\begin{enumerate}
\item Let $\gamma,\delta$ be two versions of the conditional probability of
$P$ given $\cG$.  Then the $\cG$-measurable functions
$\gamma(B,.),\delta(B,.)$ agree almost everywhere for any given $B\in\cB$,
that is, we have
\begin{equation}\label{eq.regcondae}
\forall B\in\cB:\quad P(\{\omega\;|\; \gamma(B,\omega) = \delta(B,\omega)\}) = 1.
\end{equation}
\item Conditional probabilities always exist. 
Existence of regular conditional probabilities is not assured
for arbitrary measurable spaces. However, for standard spaces
$(\Omega,\cB)$ existence can be proven. 
\item Note that it cannot be shown for arbitrary measurable spaces that two versions 
$\delta,\gamma$ agree almost everywhere \emph{for all} $B\in\cB$, meaning
that we do not have
\begin{equation*}
P(\{\omega \;|\; \forall B\in\cB : \gamma(B,\omega) = \delta(B,\omega)\}) = 1.
\end{equation*}
However, for standard spaces $(\Omega,\cB)$ this beneficial
property applies: 
\end{enumerate}
\medbreak

\begin{lemma}\label{l.standreg}
Let $(\Omega,\cB)$ be a measurable space such that $\cB$ is generated by a
countable field $\cF$. Let $P$ be a probability measure on it and assume that
the regular conditional probability of $P$ given a sub-$\sigma$-algebra $\cG$
exists. If $\delta,\gamma$ are two versions of it then the measures
$\delta(.,\omega)$ and $\gamma(.,\omega)$ agree on a set of measure one, that
is,
\begin{equation*}
P(\{\omega\; |\; \forall B\in\cB:\;\gamma(B,\omega) = \delta(B,\omega)\}) = 1.
\end{equation*}
\end{lemma}

We display the proof, as its (routine) arguments are needed in
subsequent sections.
\medbreak

\begin{proof} Enumerate the elements of $\cF$ and write $F_k$ for element
No. $k$.  According to (\ref{eq.regcondae}) we find for each $k\in\N$
a set $B_k$ of $P$-measure one on which $\delta(F_k,.)$ and
$\gamma(F_k,.)$ agree. Hence, on $B:=\bigcap_kB_k$, which is an event
of $P$-measure one, all of the $\delta(F_k,.)$ and the $\gamma(F_k,.)$
coincide. Thus the measures $\delta(.,\omega)$ and $\gamma(.,\omega)$
agree on a generating field of $\cB$ for $\omega\in B$. As a measure
is uniquely determined by its values on a generating field
(\cite{Halmos}), we obtain that the measures $\delta(.,\omega)$ and
$\gamma(.,\omega)$ agree on $B$, that is, $P$-almost
everywhere.\qed\end{proof}
\medbreak

We also give the definition of conditional expectations and point out their
extra properties on standard spaces.

\begin{definition}\label{d.condexp}
Let $(\Omega,\cB,P)$ be a probability space and $f\in L_1(P)$.
Let $\cG\subset\cB$ be a sub-$\sigma$-algebra.
If $h:\Omega\to\R$ is 
\begin{enumerate}
\item $\cG$-measurable and
\item for all $G\in\cG$ it holds that $\int_Gf\,dP = \int_Gh\,dP$
\end{enumerate}
we say that $h$ is a {\bf version of the conditional expectation}
of $f$ given $\cG$ and write
\begin{equation*}
h(\omega) = E(f|\cG)(\omega).
\end{equation*}
\end{definition}

Conditional expectations always exist. In case of standard spaces they
have an extra property which we rely on. See \cite{Partha05}, ch.~6 for
proofs of the following results.

\begin{theorem}\label{t.condexp}
Let $(\Omega,\cB,P)$ be a probability space, $\cG$ a
sub-$\sigma$-algebra of $\cB$ and $f\in L_1(P)$. Then there exists a
version $E(f|\cG)$ of the conditional expectation.  In case of a
standard space $(\Omega,\cB)$ it holds that
\begin{equation}\label{eq.version}
E(f|\cG)(\omega) = \int f(x)\,d\delta_P(x,\omega)
\end{equation}
where $\delta_P$ is a version of the regular conditional probability
of $P$ given $\cG$.
\end{theorem}

\begin{corollary}\label{c.condexp}
Let $(\Omega,\cB)$ be a standard space, $P$ a probability measure
on it and $f\in L_1(P)$. Let $\cG$ be a sub-$\sigma$-algebra and
$\delta_P$ the regular conditional probability of $P$ given $\cG$.
Then $\omega\mapsto\int f\,d\delta_P(.,\omega)$
is $\cG$-measurable (hence also $\cB$-measurable) and 
\begin{equation}\label{eq.condexp}
\int_G f\,dP = \int_G (\int f\,d\delta_P(.,\omega))\,dP
\end{equation}
for all $G\in\cG$. 
\end{corollary}

\section{Proof of Theorem \ref{t.amsdecomp}}
\label{sec.decomp}

We recall the notations of section~\ref{sec.results} and that,
according to the assumptions of theorem~\ref{t.amsdecomp}, $P$ is a
measure on a standard space $(\O,\cB)$ that is AMS relative to the
measurable $T:\O\to\O$.

\subsection{Sketch of the Proof Strategy}
\label{ssec.strategy}

The core idea for proving the theorem is to define the measures
$P_{\omega}$ as being induced by the regular conditional probability
measures of $P$ given the invariant events $\cI$. That is, we define
\begin{equation}
\label{eq.pomega}
\forall B\in\cB:\quad P_{\omega}(B):=\delta_P(B,\omega)
\end{equation}
where, here and in the following, $\delta$ refers to regular
conditional probabilities given the invariant events $\cI$.  Note
that, for arbitrary probability measures $P$ on $(\Omega,\cB)$,
\begin{equation}\label{eq.condinv}
\delta_P(B,\omega) = \delta_P(B,T\omega),
\end{equation}
as, otherwise, $\delta_P(B,.)^{-1}({y})$ would not be an invariant set
for $y:=\delta_P(B,T\omega)$ which would be a contradiction to the
$\cI$-measurability of $\delta_P(B,.)$.\\ 

As a consequence of (\ref{eq.condinv}), we obtain property $(a)$ of
the theorem. Furthermore, $(b)$ is the defining property $(CP2)$ of a
regular conditional probability (see Def.~\ref{d.regprob}) and $(c)$
is equation (\ref{eq.condexp}) from corollary~\ref{c.condexp} with
$G=\Omega$. What remains to show is that, for $\omega$ in an invariant
set $E$ of $P$-measure one, the $P_{\omega}$ are ergodic and AMS.\\

We intend to do this by the following strategy.  First, we recall that
if, in theorem \ref{t.amsdecomp}, \emph{AMS} is replaced by
\emph{stationary}, we obtain the well known result of the ergodic
decomposition of stationary measures (see the introduction for a
discussion).  If one follows the lines of argumentation of its proof
(see \cite{Gray84}, th.~2.5) one sees that, on an invariant set of
$P$-measure one, the $P_{\omega}$ are just the regular conditional
probabilities of the stationary $P$.  Applying the ergodic
decomposition of stationary measures to the stationary mean $\bar{P}$
of $P$ provides us with an invariant set $\bar{E}$ of $P$-measure $1$
such that
\begin{equation}\label{eq.barE}
\omega\in\bar{E}\quad\Longrightarrow\quad
\bar{P}_{\omega}:=\delta_{\bar{P}}(.,\omega)\text{ is stationary and ergodic.}
\end{equation}

We will show that, on an invariant set $E\subset\bar{E}$ of
$P$-measure one, the $P_{\omega}$ converge Skorokhod weakly (hence
strongly, see Def.~\ref{d.measureconv}) to the $\bar{P}_{\omega}$,
which translates to that the $P_{\omega}$ are AMS and have stationary
means $\bar{P}_{\omega}$.  As an AMS measure is ergodic if its
stationary mean is ergodic, we will have completed the proof.\\

Therefore, we will proceed according to the following steps:
\paragraph{\bf Step 1} We construct measures $Q_{\omega}$ that dominate
  $\bar{P}_{\omega}$ and all of the
  \begin{equation}\label{eq.pnomega}
  P_{n,\omega}:=\frac1n\sum_{t=0}^{n-1}(P_{\omega}\circ T^{-n}), n\ge
  0 
  \end{equation}
  (note that $P_{\omega} = P_{1,\omega}$), which will provide us
  with densities
  \begin{equation}\label{eq.omegadensities}
  f_{n,\omega} := \frac{dP_{n,\omega}}{dQ_{\omega}}
  \quad\text{ and }\quad \bar{f}_{\omega}:=\frac{d\bar{P}_{\omega}}{dQ_{\omega}}
  \end{equation}
  for all $\omega$.
\paragraph{\bf Step 2} We construct positive contractions $U_{\omega}$ on $L_1(Q_{\omega})$
  such that 
  \begin{equation}
  \label{eq.uomega}
  U_{\omega}\frac{d(P_{\omega}\circ T^{-n})}{dQ_{\omega}} =
  \frac{d(P_{\omega}\circ T^{-n-1})}{dQ_{\omega}}
  \end{equation}
  hence 
  \begin{equation}
  \label{eq.aomega}
  A_nf_{1,\omega} := \frac1n\sum_{t=0}^{n-1}U_{\omega}^tf_{1,\omega} = f_{n,\omega}
  \end{equation}
  We apply Krengel's theorem (th.~\ref{t.krengel}) to obtain that the
  $f_{n,\omega}$ converge stochastically to a $U_{\omega}$-invariant
  $f^*_{\omega}$ as well as $f^*_{\omega} =
  \liminf_{n\to\infty}f_{n,\omega}$ in $L_1(Q_{\omega})$
\paragraph{\bf Step 3} We show that, for $\omega$ in an invariant set $E$ of
  $P$-measure one,
  \begin{equation*}
  f^*_{\omega} = \bar{f}_{\omega}\quad\text{ in }L_1(Q_{\omega}).
  \end{equation*}
  This completes the proof, as this states that the $P_{\omega}$ converge
  Skorokhod weakly to the $\bar{P}_{\omega}$ in $E$, 
  hence that the $P_{\omega}$ are ergodic and AMS for $\omega$ in the 
  invariant set $E$ of $P$-measure one.

\subsection{Step 1}
\label{ssec.step1}

We recall definitions (\ref{eq.pn}) and (\ref{eq.Q}) of $P_n$ and $Q$.
We define $Q_{\omega}$ as the probability measures induced by the
regular conditional probability of $Q$ given the invariant events
$\cI$, that is,
\begin{equation}
\label{eq.Qomega}
Q_{\omega}(B) := \delta_Q(B,\omega)
\end{equation}
for $B\in\cB$.  It remains to show that, by choosing an appropriate
version, $Q_{\omega}$ indeed dominates all of the $P_{\omega}\circ
T^{-n}$ (hence all of the $P_{n,\omega}$) as well as
$\bar{P}_{\omega}$. This is established by the following lemma whose
merely technical proof has been deferred to appendix~\ref{app.lemma}.

\begin{lemma}\label{l.qversion}
\begin{equation}\label{eq.delq}
\alpha(B,\omega) := 
\frac12(\bar{P}_{\omega}(B) 
+ \sum_{n\ge 0}2^{-n-1}P_{\omega}(T^{-n}B))
\end{equation}
is a version of the regular conditional probability of $Q$ given $\cI$.
\end{lemma}
  
\begin{remark}
  In order to achieve that $Q_{\omega}$ dominates all of
  the $P_{\omega}\circ T^{-n}$ and $\bar{P}_{\omega}$ one could have
  defined $Q_{\omega}$ directly via (\ref{eq.delq}).
  However, the observation that $Q_{\omega}$
  is induced by the regular conditional probability of $Q$ given $\cI$
  is crucial for step 3.  
\end{remark}

\subsection{Step 2}

Construction of positive contractions $U_{\omega}$ on
$L_1(Q_{\omega})$ is achieved by, mutatis mutandis, reiterating the
arguments accompanying the construction of $U$ in the proof of
lemma~\ref{l.amsconv2}. In more detail, we replace
$P,P_n,\bar{P},Q,f_n,\bar{f}$ there by
$P_{\omega},\bar{P}_{\omega},P_{n,\omega},Q_{\omega},f_{n,\omega},\bar{f}_{\omega}$
(we recall
(\ref{eq.pomega}),(\ref{eq.barE}),(\ref{eq.pnomega}),(\ref{eq.Qomega}),(\ref{eq.omegadensities})
for the latter definitions) here. Note that choosing the version of
$Q_{\omega}$ according to lemma~\ref{l.qversion} ensures that
$U_{\omega}$ indeed maps $L_1(Q_{\omega})$ onto $L_1(Q_{\omega})$.\par
(\ref{eq.uomega}) and (\ref{eq.aomega}) then are a direct consequence
of the definition of $U_{\omega}$. Finally, application of Krengel's
theorem~\ref{t.krengel} to the positive contraction $U_{\omega}$ on
$L_1(Q_{\omega})$ yields a $U_{\omega}$-invariant $f^*_{\omega}$ to
which the $f_{n,\omega}$ converge stochastically.  Moreover, again by
Krengel's theorem,
\begin{equation}
\label{eq.fstar}
f^*_{\omega} = \liminf_{n\to\infty}f_{n,\omega}\quad\text{ in }L_1(Q_{\omega}).
\end{equation}

\subsection{Step 3}

We have to show that
\begin{equation*}
f^*_{\omega} = \bar{f}_{\omega}\quad\text{ in }L_1(Q_{\omega})
\end{equation*}
for $\omega$ in an invariant set $E\subset \bar{E}$ with $Q(E)=1$.  In
a first step, the following lemma will provide as with a useful
invariant $E^*$ where $E\subset E^*\subset\bar{E}$ and $Q(E^*)=1$.  We
further recall the definitions of $f_n$ and $\bar{f}$ as the densities
of $P_n$ and $\bar{P}$ w.r.t. $Q$ (see (\ref{eq.fn})).  Without loss
of generality, we choose representatives that are everywhere
nonnegative. Due to lemma \ref{l.amsconv2},
\begin{equation}\label{eq.fbarae2}
\liminf_{n\to\infty}f_n = \bar{f}\quad\text{ in }L_1(Q).
\end{equation}

\begin{lemma}\label{l.liminfae}
There is an invariant set $E^*$ with $P(E^*) = Q(E^*) = 1$ such that,
for $\omega\in E^*$,
\begin{equation}
\label{eq.fstarae}
\liminf_{n\to\infty}f_n =\liminf_{n\to\infty}f_{n,\omega}\quad\text{ in }L_1(Q_{\omega})
\end{equation}
and
\begin{equation}
\label{eq.fbarae}
\bar{f} = \bar{f}_{\omega}\quad\text{ in }L_1(Q_{\omega}).
\end{equation} 
\end{lemma}

\begin{proof} We have deferred the merely technical proof to
appendix~\ref{app.proof2}.  
\qed\end{proof} 

\noindent We compute
\begin{equation*}
\begin{split}
\int_{E^*}(\int |f^*_{\omega}-\bar{f}_{\omega}|\,dQ_{\omega})\,dQ 
&\stackrel{(\ref{eq.fstar}),(\ref{eq.fstarae}),(\ref{eq.fbarae})}{=} 
\int_{E^*}(\int |\liminf_{n\to\infty}f_n
-\bar{f}|\,dQ_{\omega})\,dQ\\
&\stackrel{(*)}{=} 
\int_{E^*} |\liminf_{n\to\infty}f_n-\bar{f}|\,dQ 
\stackrel{(\ref{eq.fbarae2})}{=} 0
\end{split}
\end{equation*}
\begin{sloppypar}
\noindent where $(*)$ follows from the defining properties of the conditional
expectation $E(|\liminf_{n\to\infty}f_n-\bar{f}|\;|\;\cI)$ in
combination with theorem~\ref{t.condexp}.  According to the last
computation, we find a set $E\subset E^*$ with $Q(E) = 1$ such that
\end{sloppypar}
\begin{equation*}
\omega\in E\quad\Longrightarrow\quad\int |f^*_{\omega}-\bar{f}_{\omega}|\,dQ_{\omega} = 0. 
\end{equation*}
The invariance of the regular conditional probabilities (see
(\ref{eq.condinv})) involved in the definitions of
$f^*_{\omega},\bar{f}_{\omega}$ implies
\begin{equation*}
\int |f^*_{\omega}-\bar{f}_{\omega}|\,dQ_{\omega} = 0\quad\Longleftrightarrow\quad
\int |f^*_{T\omega}-\bar{f}_{T\omega}|\,dQ_{T\omega} = 0.
\end{equation*}
This translates to that $E$ is invariant such that $E$ meets the requirements
of theorem~\ref{t.amsdecomp}.
\qed

\section{Discussion}\label{sec.discussion}
We have demonstrated how to decompose AMS random sources, which
encompass a large variety of sources of practical interest, into
ergodic components. The result comes in the tradition of the ergodic
decomposition of stationary sources. As outlined in the introduction,
this substantially added to source coding theory by facilitating the
generalization of a variety of prominent theorems to arbitrary, not
necessarily ergodic, stationary sources.\par Our result can be
expected to yield similar contributions to the theory of AMS
sources. An immediate clue is that the theorems developed in
\cite{Gray84} for two-sided AMS sources are now valid for arbitrary
AMS sources by replacing theorem 2.6 there by
theorem~\ref{t.amsdecomp} here.\par Moreover, a couple of relevant
quantities in information theory (e.g.~entropy rate) are affine
functionals that are upper semicontinuous w.r.t.~the space of
stationary random sources, equipped with the weak topology.  Jacobs'
theory of such functionals (\cite{Jacobs63}, see also \cite{Effros94},
th.~4) immediately builds on the ergodic decomposition of stationary
sources. This theory should now be extendable to AMS sources.\par We
finally would like to mention that a certain class of source coding
theorems for AMS sources were obtained by partially circumventing the
lack of an ergodic decomposition. Schematically, this was done by a
reduction from AMS sources to their stationary means and subsequent
application of the ergodic decomposition for stationary sources in
order to further reduce to ergodic sources. In these cases, our
contribution would only be to simplify the theorems' statements and
thus a merely esthetical one. However, in the remaining cases where
the reduction from asymptotic mean stationarity to stationarity is not
applicable, our result will be essential. The full exploration of
related consequences seems to be a worthwhile undertaking.

\section{Acknowledgments}
The author would like to thank the Pacific Institute for the
Mathematical Sciences for funding.

\begin{appendix}

\section{Proof of lemma~\ref{l.qversion}}
\label{app.lemma}

In the following, according to the assumptions of
theorem~\ref{t.amsdecomp}, $P$ is a measure on a standard space
$(\O,\cB)$ that is AMS relative to the measurable $T:\O\to\O$. 
We further recall the notations of section~\ref{sec.results} 
as well as equations (\ref{eq.pn}) and (\ref{eq.Q}) for the necessary
definitions.

\begin{lemma}\label{l.invfunc}
Let $g:\Omega\to\R$ be a
$T$-invariant (that is, $g(\omega) = g(T\omega)$ for all
$\omega\in\O$), measurable function. Then it holds that
\begin{equation}\label{eq.invfunc}
\int g\,dP = \int g\,d(P\circ T^{-n}) = \int g\,dP_n = \int g\,d\bar{P} = \int g\,dQ.
\end{equation}
In particular, all of the integrals exist if one of the integrals exists.
\end{lemma}

\begin{proof} Note that $Q$ and all of the $P\circ T^{-n}$ and $P_n$, like $P$,
are AMS with stationary mean $\bar{P}$, which is an obvious
consequence of their definitions. Therefore, the claim of the lemma
follows from the, intuitively obvious, observation that $\int g\,dP =
\int g\,d\bar{P}$ for invariant $g$ and general AMS $P$ with
stationary mean $\bar{P}$. See \cite{Gray01} for
details.\qed\end{proof}
\medbreak

\begin{lemma}\label{l.tversion}
The functions 
\begin{equation*}
\zeta_n(B,\omega) := \delta_P(T^{-n}B,\omega) = P_{\omega}(T^{-n}B)
\end{equation*}
are versions of the regular conditional probabilities $\delta_{P\circ
  T^{-n}}$ of the $P\circ T^{-n}$ given $\cI$.
\end{lemma}

\begin{proof} The functions $\zeta_n(.,\omega)$ are probability
measures for fixed $\omega\in\Omega$ (this is $(RCP)$ of definition
\ref{d.regprob}) as the $P_{\omega}$ are, by the definition of
$\delta_P$. Again by the definition of $\delta_P$, $\zeta_n(B,.)$ is
also $\cI$-measurable in $\omega$ for fixed $B\in\cB$.  which is
$(CP1)$ of definition \ref{d.regprob}.  For $I\in\cI$ and $B\in\cB$ we
compute
\begin{equation*}
\begin{split}
\int_I\delta_P(T^{-n}B,\omega)\,d(P\circ T^{-n})(\omega) 
&\stackrel{(\ref{eq.condinv}),(\ref{eq.invfunc})}{=} \int_I\delta_P(T^{-n}B,\omega)\,dP(\omega)\\
&= P(I\cap T^{-n}B) \stackrel{T^{-n}I = I}{=} P(T^{-n}(I\cap B))\\ 
&= \int_I\delta_{P\circ T^{-n}}(B,\omega)\,d(P\circ T^{-n})(\omega) 
\end{split}
\end{equation*}
where the first equation follows from the invariance of
the integrands and lemma \ref{l.invfunc}. We have thus shown
$(CP2)$ of definition \ref{d.regprob}.\qed\end{proof}
\medbreak
We recall that, for lemma~\ref{l.qversion}, we have to show that
\begin{equation*}
\alpha(B,\omega) = \frac12(\bar{P}_{\omega}(B)
+ \sum_{n\ge 0}2^{-n-1}P_{\omega}(T^{-n}B))
\end{equation*}
is a version of the regular conditional probability $\delta_Q$.  Note
first that $\bar{P}_{\omega}$, according to our proof strategy
outlined in subsection~\ref{ssec.strategy}, was defined as
$\delta_{\bar{P}}(.,\omega)$ where $\delta_{\bar{P}}$ is the regular
conditional probability of the stationary mean $\bar{P}$.
Furthermore, as a consequence of lemma~\ref{l.tversion}, we can
identify the $P_{\omega}\circ T^{-n}$ with $\delta_{P\circ
  T^{-n}}(.,\omega)$ and write
\begin{equation}
\label{eq.deltastar}
\alpha(B,\omega) = 
\frac12(\delta_{\bar{P}}(B,\omega) 
+ \sum_{n\ge 0}2^{-n-1}\delta_{P\circ T^{-n}}(B,\omega)).
\end{equation}
We will then exploit the defining properties of the $\delta$s to
finally show that $\alpha$ is a version of $\delta_Q$.\\

\noindent {\it Proof of lemma~\ref{l.qversion}.$\,$} We have to
check properties $(RCP), (CP1)$ and $(CP2)$ of definition
~\ref{d.regprob}.\\

$(RCP)\;$: That $\alpha(.,\omega)$ is a probability measure for fixed
$\omega$ follows from an argumentation which is completely analogous
to that at the beginning of section~\ref{sec.amsconv}, surrounding
equations (\ref{eq.Qa}) and (\ref{eq.Q}).\\

$(CP1)\;$: As all of the $\delta$'s involved in (\ref{eq.deltastar})
are invariant in $\omega$ (see (\ref{eq.condinv})), we know that
$\alpha(B,.)$ is measurable w.~r.~t.~$\cI$ for any $B\in\cB$ which is
$(CP1)$ of definition ~\ref{d.regprob}.\\

$(CP2)\;$: Fix $B\in\cB$ and consider the functions
\begin{equation*}
g_n(\omega):=\frac12(\delta_{\bar{P}}(B,\omega)+\sum_{k=0}^n2^{-k-1}\delta_{P\circ T^{-k}}(B,\omega)).
\end{equation*}
This is an increasing sequence of non-negative measurements which
converges everywhere to the values $\alpha(B,\omega)$. Because of
(\ref{eq.condinv}) the summands of $g_n$ are invariant. As all of the
summands are also integrable with respect to some $P\circ T^{-k}$ or
$\bar{P}$ they are also integrable with respect to $Q$, due to
lemma~\ref{l.invfunc}. Therefore, also the $g_n$ are integrable with
respect to $Q$. The monotone convergence theorem of Beppo Levi
(e.g. \cite{Halmos}) reveals that also $\alpha(B,.)$ is and further,
for $I\in\cI$ and $B\in\cB$:
\begin{equation*}
\begin{split}
\int_I \alpha(B,\omega)\,dQ(\omega) 
&= \int_I \lim_{n\to\infty}\frac12(\delta_{\bar{P}}(B,\omega)
+\sum_{k=0}^n2^{-k-1}\delta_{P\circ T^{-k}}(B,\omega))\,dQ(\omega)\\
&\stackrel{(a)}{=}\lim_{n\to\infty}\int_I \frac12(\delta_{\bar{P}}(B,\omega)
+\sum_{k=0}^n2^{-k-1}\delta_{P\circ T^{-k}}(B,\omega))\,dQ(\omega)\\
&\stackrel{(b)}{=}\lim_{n\to\infty}
\frac12(\int_I\delta_{\bar{P}}(B,\omega)\,d\bar{P}(\omega)\\
&\qquad\qquad\;\; +\sum_{k=0}^n2^{-k-1}\int_I\delta_{P\circ T^{-k}}(B,\omega)\,d(P\circ T^{-k})(\omega)) \\
&\stackrel{(c)}{=} \lim_{n\to\infty}\frac12(\bar{P}(I\cap B) + \sum_{k=0}^n2^{-k-1}P(T^{-k}(I\cap B))\\
&= \frac12(\bar{P}(I\cap B) + \sum_{n\ge 0}2^{-n-1}P(T^{-n}(I\cap B)))\\
&= Q(I\cap B)
\end{split}
\end{equation*}
where $(a)$ follows from Beppo Levi's theorem, $(b)$ follows from the
invariance of the $\delta$s and subsequent application of
lemma~\ref{l.invfunc} and $(c)$ is just the defining property $(CP2)$
of the conditional probabilities $\delta$
(definition~\ref{d.regprob}). We thus have shown property $(CP2)$ for
$\alpha$.\qed
\medbreak

\section{Proof of Lemma~\ref{l.liminfae}}
\label{app.proof2}

According to the assumptions of theorem~\ref{t.amsdecomp}, $P$ is a
measure on a standard space $(\O,\cB)$ that is AMS relative to the
measurable $T:\O\to\O$.  We further recall the notations of
section~\ref{sec.results} as well as equations (\ref{eq.pn}),
(\ref{eq.Q}), (\ref{eq.fn}), (\ref{eq.pomega}), (\ref{eq.barE}),
(\ref{eq.pnomega}), (\ref{eq.omegadensities}), (\ref{eq.Qomega}) and
the surrounding texts for the necessary definitions. We further remind
that, without loss of generality, we had chosen representatives of the
$f_n$ and $\bar{f}$ that are everywhere nonnegative. The following
lemma will deliver the technical key to lemma~\ref{l.liminfae}.

\begin{lemma}\label{l.version}
For each $1\le n\in\N$ there is an invariant $E_n\in\cI\subset\cB$ with $P_n(E_n) = Q(E_n) = 1$
such that 
\begin{equation*}
\omega\in E_n\quad\Longrightarrow\quad f_{n,\omega} = f_n\quad\text{ in }L_1(Q_{\omega}).
\end{equation*}
There is also an invariant $E_{\infty}$ with $\bar{P}(E_{\infty}) = Q(E_{\infty}) = 1$
such that
\begin{equation*}
\omega\in E_{\infty}\quad\Longrightarrow\quad \bar{f}_{\omega} = \bar{f}\quad\text{ in }L_1(Q_{\omega}).
\end{equation*}
\end{lemma}
\medskip

Loosely speaking, the lemma reveals that the $f_n$ and the
$f_{n,\omega}$ as well as $\bar{f}$ and $\bar{f}_{\omega}$ agree
$Q_{\omega}$-a.e, for $Q$-almost all $\omega\in\O$.  This
means that, for $Q$-almost all $\omega$, they are equal on the parts
of $\Omega$ considered relevant by the measures $Q_{\omega}$.

\begin{proof}
Consider the functions
\begin{equation*}
\beta_n(B,\omega) := \int_Bf_{n,\omega}\,dQ_{\omega}\quad\text{ and }\quad
\gamma_n(B,\omega) := \int_Bf_n\,dQ_{\omega}
\end{equation*}
By the definition of a density, 
\begin{equation*}
\delta_{P_n}(B,\omega) = \int_Bf_{n,\omega}\,dQ_{\omega}.
\end{equation*}
Hence $\beta_n(B,\omega)$ is just the regular conditional
probability of $P_n$ given $\cI$. We now show that $\gamma_n$ is a
version of the conditional probability of $P_n$ given $\cI$ (but not
necessarily a regular one). Note first that the $\gamma_n(B,.)$ are
$\cI$-measurable as, according to (\ref{eq.version}), we have that
$\gamma_n(B,\omega)$ agrees with the conditional expection
$E_Q(\mathbf{1}_Bf_n|\cI)(\omega)$, which, by definition, is
$\cI$-measurable. Second, we observe that, for $I\in\cI$ and
$B\in\cB$, as $\gamma_n$ is invariant in $\omega$ $(*)$,
\begin{equation*}
\begin{split}
\int_I\gamma_n(B,\omega)\,dP_n &\stackrel{(*),(\ref{eq.invfunc})}{=} \int_I \gamma_n(B,\omega)\,dQ\\
 &= \int_I (\int_Bf_n\,dQ_{\omega})\,dQ\\
&= \int_I (\int\mathbf{1}_Bf_n\,dQ_{\omega})\,dQ\\
&\stackrel{(\ref{eq.condexp})}{=} \int_I \mathbf{1}_Bf_n\,dQ = \int_{I\cap B}f_n\,dQ\\
&= P_n(I\cap B),
\end{split}
\end{equation*}
which shows the required property $(CP2)$ of
definition~\ref{d.regprob}. Hence the $\gamma_n$'s are versions of the
conditional probabilities of the $P_n$'s given $\cI$.
\medbreak

Note that the $\gamma_n(.,\omega)$ are measures because the $f_n$ had
been chosen nonnegative everywhere.
If we follow the line of argumentation of lemma~\ref{l.standreg}
we find a set $E_n$ of $P_n$-measure one such that the
measures $\beta_n(.,\omega)$ and $\gamma_n(.,\omega)$ agree for
$\omega\in E_n$. Because of the invariance of $\beta_n,\gamma_n$ the
set $E_n$ is invariant. Hence (lemma~\ref{l.invfunc}) also $Q(E_n) = 1$.
Resuming we have
\begin{equation*}
\omega\in E_n\quad\Longrightarrow\quad 
\forall B\in\cB:\;\int_Bf_n\,dQ_{\omega} = \int_Bf_{n,\omega}\,dQ_{\omega}.
\end{equation*}
As two functions agree almost everywhere if their integrals conincide
over arbitrary events, we are done with the assertion of the lemma for 
the $f_n$.
\medbreak

We find an invariant set $E_{\infty}$ with $\bar{P}(E_{\infty}) = Q(E_{\infty}) = 1$ 
such that
\begin{equation*}
\bar{f}_{\omega}=\bar{f}\quad\text{ in }L_1(Q_{\omega})
\end{equation*}
for $\omega\in E_{\infty}$
by a
completely analogous argumentation .\qed\end{proof}
\medbreak

\noindent {\it Proof of lemma~\ref{l.liminfae}.$\,$} Define
\begin{equation}
\label{eq.estar}
E^*:=E_{\infty} \cap (\bigcap_{n\ge 1}E_n) 
\end{equation}
with $E_{\infty}$ and the
$E_n$ from lemma~\ref{l.version}. $E^*$ is invariant and
$Q(E^*) = 1$ as it applies to all sets on the right hand side of (\ref{eq.estar}). We obtain
\begin{equation*}
\forall n\in\N\;\; f_n = f_{n,\omega}\quad\text{ and }\quad\bar{f} = \bar{f}_{\omega}\quad\text{ in }L_1(Q_{\omega})
\end{equation*}
for $\omega\in E^*$. Therefore also 
\begin{equation*}
\liminf_{n\to\infty}f_n 
=\liminf_{n\to\infty}f_{n,\omega} \quad\text{ in }L_1(Q_{\omega})
\end{equation*}
for $\omega\in E^*$.\qed
\medbreak

\end{appendix}

\end{document}